\documentclass{elsarticle}

\usepackage{lineno,hyperref}
\usepackage{amsmath,amssymb,amsthm}
\usepackage{color}
\usepackage{algpseudocode}
\usepackage{setspace}
\usepackage{graphicx}
\usepackage{relsize}
\usepackage{arydshln}
\usepackage{tikz}
\usetikzlibrary{arrows,shapes,backgrounds,calc,positioning}
\modulolinenumbers[5]

\journal{Applied Mathematics and Computation}

\definecolor{darkgreen}{rgb}{0.0,0.7,0.0}

\definecolor{darkred}{rgb}{0.75,0.0,0.0}

\newcommand{\R}{\mathbb{R}}

\theoremstyle{plain}
\newtheorem{thm}{Theorem}[section]

\theoremstyle{definition}
\newtheorem{defn}{Definition}[section]

\newtheorem{rem}{Remark}[section]
\begin{document}

\begin{frontmatter}

\title{Numerical optimal control for HIV prevention with dynamic budget allocation\tnoteref{mytitlenote}}
\tnotetext[mytitlenote]{The work of O.~S. Serea was partially supported by The French National Research Agency, Project ANR-10-BLAN 0112}

\author[Addr1]{Dmitry Gromov\corref{corrauth}}
\cortext[corrauth]{Corresponding author}
\ead{dv.gromov@gmail.com}

\author[Addr2]{Ingo Bulla}
\ead{ingobulla@gmail.com}

\author[Addr3]{Ethan O. Romero-Severson}
\ead{eoromero@lanl.gov}

\author[Addr4]{Oana Silvia Serea}
\ead{oana-silvia.serea@univ-perp.fr}

\address[Addr1]{Faculty of Applied Mathematics and Control Processes, Saint Petersburg State University, St. Petersburg, Russia}
\address[Addr2]{Institut f\"ur Mathematik und Informatik, Walther-Rathenau-Stra{\ss}e 47, 17487 Greifswald, Germany}
\address[Addr3]{Theoretical Biology and Biophysics Group, Los Alamos National Laboratory, Los Alamos, New Mexico, USA}
\address[Addr4]{ Univ. Perpignan Via Domitia, Laboratoire de
Math\'{e}matique et Physique, EA 4217, F-66860 Perpignan, France}

\begin{abstract}
This paper is about numerical control of HIV propagation. The contribution of the paper is threefold: first, a novel model of HIV propagation is proposed; second, the methods from numerical optimal control are successfully applied to the developed model to compute optimal control profiles; finally, the computed results are applied to the real problem yielding important and practically relevant results.
\end{abstract}

\begin{keyword}
HIV \sep AIDS \sep Highly active antiretroviral therapy \sep Treatment as prevention \sep Pre-exposure prophylaxis \sep Optimal control \sep Resource allocation
\end{keyword}

\end{frontmatter}

\linenumbers

\section{Introduction}

Since the outbreak of the global HIV/AIDS pandemic in the early 1980s about 35 million people died from AIDS-related illnesses \cite{WHO_HIV}. 
Although the annual number of new cases of HIV (incidence) has been decreasing globally, among men who have sex with men (MSM) in developed countries incidence is increasing \cite{Beyrer:16, Hall:08}.
The current public health challenge in these populations is effective triage of limited prevention resources. 
Two recent advancements in HIV prevention science are relevant to this issue. 
First, new forms of highly active antiretroviral therapy (HAART) are so effective that viral load in infected patients become undetectable and that patients on effective treatment are essentially non-infectious \cite{Donnell:10}; we will refer to this as the Treatment as Prevention (TaP) modality.
Second, new treatments targeted at uninfected persons have been shown to reduce the rate of infection to nearly zero when treatment is adhered to \cite{Baeten:12}, which is referred to as Pre-Exposure Prophylaxis (PrEP).
We have therapies that are effective at both blocking transmissions from infected persons and protecting uninfected persons from becoming infected; however, there is little consensus on how these technologies should be deployed \cite{cohen_hiv_2012, Eaton:12}.
The main obstacle is that while the individual-level efficacy of these interventions can be ascertained with controlled trials, their effectiveness as public health interventions cannot.   
This is due to the fact that populations are different from one another both in terms of the fundamental transmission dynamics of HIV, but also in the availability of prevention resources.
Identifying the optimal strategy of resource allocation must be based on a model of the underlying medical, biological, and social processes that captures the relevant features of the population.

Theoretical studies of intervention effectiveness are generally based on either state-space models represented by a system of ordinary differential equations (ODE) (often referred to as `compartmental` models in the epidemiology literature, \cite{Jacquez:93,Haddad:10}) and agent-based models where large-populations of individuals with complex behavioral patters are directly simulated.
The trade-off between these approaches exchanges verisimilitude in the agent-based formation for efficient computation in the compartmental formulation.
Due to the computational challenges in calculating optimal allocation strategies, compartmental models are generally preferred \cite{Kerr:15, Shattock:16, Punyacharoensin:16,Griffiths:00}. 
In this paper, we present a method for calculating the optimal allocation strategy for a hypothetical health agency in a major US city, which receives funding to set up a long-term HIV prevention program for MSM in that city.
Our analysis is based on a compartmental model of HIV transmission that includes natural history of infection, dynamic risk behavior, and partner preference components. 
For the sake of simplicity we restrict the interventions the health agency can allocate resources to TaP and PrEP. 
This in turn allows us to address the important question to which extent the investment into PrEP for MSM should be scaled up in the future, which is an open question in HIV prevention policy.
While this type of model is in widespread use, our approach to modeling the enrollment of infecteds and high-risk susceptible into ART and PrEP, respectively, and our ability to find dynamic optimal resource allocation patterns over long time horizons is unique.

The described problem of determining the optimal allocation strategy can be formulated as an optimal control problem with a specific set of constraints. 
We developed an efficient numerical scheme to deal with this problem. 
The described approach is very general and can be applied to a wide class of optimization problems. 
To facilitate the reuse of the proposed scheme we provide a detailed description of all steps and indicate possible extensions and ramifications of the method. 
When solving optimal control problems, there are two classes of methods: {\em indirect} and {\em direct} ones. 
Indirect methods employ the first-order optimality condition to formulate a two-point boundary-value problem for the system of $2N$ Hamiltonian equation, \cite{Pontr:62}, which is solved, either analytically or numerically (see, e.g., \cite{Pyt:99,SubZbi:09}). 
Analytical solutions to these problems are available only for systems of low order and will not be treated here. The use of numerical schemes is however also restricted. 
Numerical implementations of indirect methods are often numerically unstable and require a good initial guess which is in most cases not available, \citep{Asch:94}. 
Furthermore, many non-standard constraints cannot be addressed within the framework of the classical optimal control theory. This is in particular true for the problem considered in the paper. Thus one has to resort to direct methods.

Direct methods attempt to directly minimize the cost function using  constrained nonlinear programming. Within the class of direct methods, there are different possible approaches, most notable are {\em shooting} methods and {\em simultaneous} methods (see \cite{Rao:09s,Betts:10}).
Shooting methods parametrize the controls and use numerical simulation to obtain the solution of the system's equations on the whole interval (single shooting) or on a number of subintervals (multiple shooting), \cite{StoBul:93}. Nonlinear programming is then used to optimize the control parameters while obeying the constraints. The multiple shooting method has proven to be very efficient for solving many practical problems. However, the necessity of numerical integration of system's equations makes this method very time consuming.
Simultaneous methods parametrize both the controls and the system's trajectory and determine the missing values of parameters by solving a (typically) large system of nonlinear algebraic equations. 
All unknown parameters are determined {\em simultaneously} thus giving the name to this class of methods. Along with system's equations, one describes the constraints using the introduced parametrization. This is referred to as the {\em transcription} procedure. The optimal control problem is thus formulated as a large scale nonlinear optimization problem, \cite{Garg:10}. 

A modification of the latter method was employed in the paper. It was used it to compute a set of optimal allocation strategies for a particular scenario of HIV propagation for different parameters of the treatment. 

The paper is organized as follows: in Section \ref{sec:model}, a model of HIV propagation is derived and analyzed, Section \ref{sec:OC-problem} describes the optimal control problem of resource allocation for HIV treatment and prevention while Section \ref{sec:OC-num} describes the numerical scheme used to solve the optimal control problem. Finally, Section \ref{sec:results} presents and analyzes the numerical results.

\section{Epidemiological model}\label{sec:model}

\subsection{Derivation of the model}

We base our approach on a population balance model. This means that we divide the whole population of MSM in the respective city into a number of groups. All individuals within a given group are assumed to be identical in their evolution. The state variables of the model correspond to the number of people within each group. These are described in Table \ref{tab:notation}, together with other important notation used throughout the paper. The dynamics of each state variable can be described by the following differential equation:
$$\dot{Z}=V_{in} - V_{out},$$ 
where $V_{in}$ and $V_{out}$ describe the in- and out-flows. That is the number of people that enter or leave the respective group within the unit time interval. The disaggregation is based on whether an individual is infected and in which stage, what is his risk behavior, and whether he receives treatment and which one. In our model, the time unit is set equal to [month]. 

\begin{table}[htb]
\begin{tabular}{ll}
\hline
Size variables:\\
$n$&Dimension of the state vector\\
$n_{int}$& Number of solution segments\\
$n_{cp}$& Number of collocation points\\
$m$& Dimension of the control vector\\[2pt]
\hline
Points:\\
$t_i, \; i=0,\dots,n_{int}$&Knot points\\
$\tau^i_k,\; k=0,\dots,n_{cp}$& Grid points\\[2pt]
\hline
State vectors:\\
$X(t)$&State evolution of the system,\\
$\hat{X}(t)$&Polynomial interpolation of $X(t)$ through the grid points $\tau^i_k$,\\
$\tilde{X}^i(t)$& The solution of the uncontrolled system on the $i$th interval,\\
$\mathbf{X}^i$&Matrix of the state values at grid points $\tau^i_k$\\[2pt]
\hline
State variables:&\\
$S_{\cdot}(t)$& untreated susceptible individuals\\
$I_{\cdot\, \cdot}(t)$& untreated infected individuals\\
$T_{\cdot}(t)$& infected individuals on TaP\\
$P(t)$& susceptible individuals on PrEP\\
$D(t)$ & individuals deceased due to AIDS\\
$N(t)$ & all individuals\\
$N_\cdot(t)$ & all individuals displaying a given risk behavior\\
\hline
Indices:&\\
$A/C$&Infectious stage (acute/chronic)\\
$H/L$&Risk status (high/low)\\
\hline
\end{tabular}
\caption{Notation used in the paper. The state variables have up to two indices. For variables with one index it indicates the risk status, for variables with two indices infectious stage and risk status are indicated.}
\label{tab:notation}
\end{table}

The structure of the flows within the system is shown in Fig.~\ref{fig:dynamics}, the model parameters are summarized in Table \ref{tab:param}. Transitions between states happen due to individuals becoming infected, progressing from acute to chronic stage, and dying of AIDS ($S_\cdot \rightarrow I_{A\cdot} \rightarrow I_{C\cdot} \rightarrow D$). Moreover, individuals change their risk behavior ($(\cdot)_L \leftrightarrow (\cdot)_H$, $P \rightarrow S_L$), are put on PrEP or TaP treatment or cancel it ($S_H \leftrightarrow P$, $I_{C\cdot} \leftrightarrow T_\cdot$). Finally there is flow into the system due to individuals reaching an age of sexual activity and outflow from the system because of non-HIV related death or individuals becoming sexual inactive or settling in a monogamous, lifelong relationship.

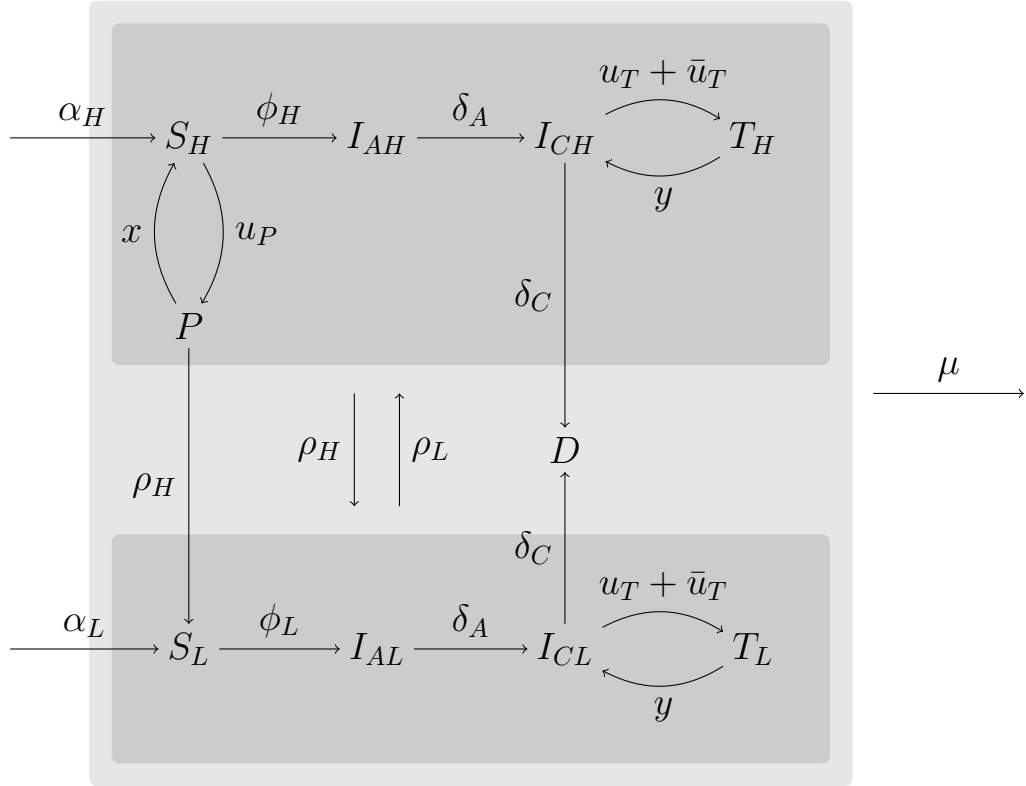
\begin{figure}[htb]
\def\shownum{0} 
\center
\begin{tikzpicture}[node distance=2.5cm, font=\relsize{2}]
\begin{scope}
  \node (X_H) {};
	\node (S_H) [right of= X_H] {$S_H$};
	\node (I_AH) [right of= S_H] {$I_{AH}$};
	\node (I_CH) [right of= I_AH] {$I_{CH}$};
	\node (P) [below of= S_H] {$P$};
	\node (T_H) [right of= I_CH] {$T_H$};
	
	\node (top_left_H) at ($(S_H) + (-8mm,13mm)$) {};
	\node (bottom_right_H) at ($(T_H |- P) + (8mm,-3mm)$) {};
	
	\draw [->] (X_H) to node[above=1pt] {$\alpha_H$} (S_H);
	\draw [->] (S_H) to node[above=1pt] {$\phi_H$} (I_AH);
	\draw [->] (S_H) to [bend left=30] node[right=1pt] {$u_P$} (P);
	\draw [->] (I_AH) to node[above=1pt] {$\delta_A$} node[below=1pt] {\ifthenelse{\shownum=1}{0.17}{}} (I_CH);
	\draw [->] (I_CH) to [bend left=30] node[above=1pt] {$u_T + \bar{u}_T$} node[below=1pt] {\ifthenelse{\shownum=1}{$\sim 0.01$}{}} (T_H);
	\draw [->] (P) to [bend left=30] node[left=1pt] {$x$} (S_H);
	\draw [->] (T_H) to [bend left=30] node[below=1pt] {$y$} (I_CH);
\end{scope}

\begin{scope}[yshift=-6.8cm]
    \node (X_L) {};
	\node (S_L) [right of= X_L] {$S_L$};
	\node (I_AL) [right of= S_L] {$I_{AL}$};
	\node (I_CL) [right of= I_AL] {$I_{CL}$};
	\node (T_L) [right of= I_CL] {$T_L$};
	
	\node (top_left_L) at ($(S_L) + (-8mm,13mm)$) {};
	\node (bottom_right_L) at ($(T_L) + (8mm,-13mm)$) {};

	\draw [->] (X_L) to node[above=1pt] {$\alpha_L$} (S_L);
	\draw [->] (S_L) to node[above=1pt] {$\phi_L$} (I_AL);
	\draw [->] (I_AL) to node[above=1pt] {$\delta_A$} node[below=1pt] {\ifthenelse{\shownum=1}{0.17}{}} (I_CL);
	\draw [->] (I_CL) to [bend left=30] node[above=1pt] {$u_T + \bar{u}_T$} node[below=1pt] {\ifthenelse{\shownum=1}{$\sim 0.01$}{}} (T_L);
	\draw [->] (T_L) to [bend left=30] node[below=1pt] {$y$} (I_CL);
\end{scope}

\node (top_left) at ($(top_left_H) + (-3mm,3mm)$) {};
\node (bottom_right) at ($(bottom_right_L) + (3mm,-3mm)$) {};

\node (dummy) at ($0.5*(bottom_right_H) + 0.5*(top_left_L)$) {};
\node (D) at ($(I_CH |- dummy)$) {$D$};

\draw [->] (I_CH) to node[left=1pt] {$\delta_C$} node[right=1pt] {\ifthenelse{\shownum=1}{$8.5\cdot10^{-3}$}{}} (D);
\draw [->] (I_CL) to node[left=1pt] {$\delta_C$} node[right=1pt] {\ifthenelse{\shownum=1}{$8.5\cdot10^{-3}$}{}} (D);
\draw [->] (P) to node[left=1pt] {$\rho_H$} (S_L);

\draw [->] ($(I_AH |- bottom_right_H) + (-3mm,-6mm)$) to node[left=1pt] {$\rho_H$} ($(I_AH |- top_left_L) + (-3mm,6mm)$);
\draw [->] ($(I_AH |- top_left_L) + (3mm,6mm)$) to node[right=1pt] {$\rho_L$} ($(I_AH |- bottom_right_H) + (3mm,-6mm)$);

\draw [->] ($0.5*(T_H) + 0.5*(T_L) + (16mm,0)$) to node[above=1pt] {$\mu$} node[below=1pt] {\ifthenelse{\shownum=1}{$2.8\cdot10^{-3}$}{}} ($0.5*(T_H) + 0.5*(T_L) + (36mm,0)$);

\begin{pgfonlayer}{background}
  \filldraw [line width=2mm,join=round,black!10]
    (top_left.north -| bottom_right.east) rectangle (bottom_right.south -| top_left.west);
  \filldraw [line width=2mm,join=round,black!20]
    (top_left_H.north -| bottom_right_H.east) rectangle (bottom_right_H.south -| top_left_H.west)
    (top_left_L.north -| bottom_right_L.east) rectangle (bottom_right_L.south -| top_left_L.west);
\end{pgfonlayer}

\end{tikzpicture}
\caption{The dynamics of the system. In the upper darkgrey box all high-risk states are located, in the lower darkgrey box all low-risk states. All annotations on the arrows are relative transition rates except for $u_P$ and $u_T$ which are non-rate quantities that govern the rates of the respective transition. There are transitions between high- and low-risk in both directions for susceptibles, infecteds, and people treated with TaP. But people treated with PrEP only can adapt a low risk behavior and become low-risk susceptibles, while there are no transitions from low risk individuals into the group on PrEP treatment. The outflow $\mu$ applies to all groups equally.} 
\label{fig:dynamics}
\end{figure}

\begin{table}[htb]
\begin{tabular}{ll}
\hline
Infection parameters:\\
$\phi_H$, $\phi_L$ & transmission rate for high-risk resp.~low-risk susceptibles \\
$\delta_A$ & rate that acutely infected individuals become chronically infected \\
$\delta_C$ & rate that chronically infected individuals die due to AIDS \\
$\lambda_H$, $\lambda_L$ & contact rate of high-risk resp.~low risk individuals \\
\hdashline
$\beta_A$, $\beta_C$ & infection probability per act for acutely resp.~chronically infecteds \\
$\pi$ & probability that a sexual contact takes place at the respective preferred sites \\
\hline
Treatment parameters:\\
x & rate at which PrEP fails or is canceled \\
$\bar{u}_T$ & baseline enrollment rate into TaP \\
y & rate at which TaP fails or is canceled \\
\hline
Other parameters:\\
$\alpha_H$, $\alpha_L$ & recruitment rate of new high-risk resp.~low-risk susceptibles \\
$\mu$ & rate that adults die of non-HIV related causes, reach an age of \\
& sexual inactivity, or settle in a monogamous, lifelong relationship \\
$\rho_H$ & rate that high-risk persons become low-risk \\
$\rho_L$ & rate that low-risk persons become high-risk \\
\hline
\end{tabular}
\caption{Model parameters. The variables are grouped into the ones directly governing the infection process and the treatment process, respectively, and other variables. The parameters $\beta_A$, $\beta_C$, and $\pi$ are all probabilities, all other parameters are rates, measured in individuals per month.}
\label{tab:param}
\end{table}

Parameters used in the model were selected from studies of general MSM populations from the United States. Behavioral parameters including the proportion of the population that is high-risk, the contact rate ratio of high and low-risk individuals, and the rate at which high-risk individuals become low risk and vice versa were taken from an analysis of longitudinal sexual contact rate data of unaffected gay men in 3 large cities in the United States \cite{romero-severson_dynamic_2015}. 

Finally, we selected values of $\lambda_L$ and $\bar{u}_T$ such that the endemic equilibrium prevalence is about 20\% and about 25\% of infected individuals are on effective treatment, which is consistent with an average large MSM population in the United States \cite{_prevalence_2010, rosenberg_modeling_2014}.   However, there are several important aspects of the model which deserve particular attention. These are described below.

\paragraph{Mixing}
To model infection events (which take place at rate $\phi_H$ resp.~$\phi_L$ in our model), we assume that transmissions occur at three distinct sites:
\begin{itemize}
\item locations frequented exclusively by high-risk individuals
\item locations frequented exclusively by low-risk individuals
\item locations jointly used by both low- and high-risk individuals
\end{itemize}
The total rate of contact of high- resp.~low-risk individuals is denoted by $\lambda_H$ resp.~$\lambda_L$. Hereby, both risk groups make a proportion $\pi$ of their contacts at the site which is exclusively frequented by their own risk group and the remainder of contacts at the mixing site. 

To derive $\phi_H$ and $\phi_L$, we omit the time dependency in the state variables for now. Then, the total number of contacts made at a given time is given by 
$$
\theta = \lambda_H N_H + \lambda_L N_L
$$
with 
\begin{eqnarray*}
N_H & = & S_H + I_{AH} + I_{CH} + P + T_H, \\
N_L & = & S_L + I_{AL} + I_{CL} + T_L.
\end{eqnarray*}
The probability that, given a random individual has a contact, this contact is with a high resp.~low risk individual is
\begin{equation*}
\eta_H = \frac{\lambda_H N_H}{\theta},\quad {\rm resp.}\quad
\eta_L = \frac{\lambda_L N_L}{\theta}.
\end{equation*}
Then, denoting the probability of transmission per contact by $\beta_A$ for acute-stage infecteds and $\beta_C$ for chronic-stage infecteds, the probability of transmission per contact at the mixing site is
\begin{eqnarray*}
\sigma & = & \beta_A \left(\eta_H \frac{I_{AH}}{N_H}+ 
                		   \eta_L \frac{I_{AL}}{N_L} \right) + 
             \beta_C \left(\eta_H \frac{I_{CH}}{N_H}+ 
                    	   \eta_L \frac{I_{CL}}{N_L} \right) \\
       & = & \theta^{-1} 
             \left[ \beta_A(\lambda_H I_{AH} + \lambda_L I_{AL}) + 
                    \beta_C(\lambda_H I_{CH} + \lambda_L I_{CL}) \right].
\end{eqnarray*}
Consequently, the per-capita instant rate of a high-risk susceptible becoming infected at the common site is
$$
\tau_H = (1-\pi)\lambda_H \sigma.
$$
Moreover, the per-capita instant rate of a high-risk susceptible becoming infected at the high-risk preferred site is
$$ 
\psi_H = \pi \lambda_H\left[\beta_A \frac{I_{AH}}{N_H} + \beta_C \frac{I_{CH}}{N_H}\right].
$$
Therefore, the total per-capita transmission rate for high-risk susceptibles is given by
$$
 \phi_H = \psi_H + \tau_H.
$$
Analogously, the per-capita instant rate of a low-risk susceptible becoming infected at the common site is
$$
\tau_L = (1-\pi)\lambda_L \sigma
$$
and the per-capita instant rate of a low-risk susceptible becoming infected at the low-risk preferred site is
$$
\psi_L = \pi \lambda_L\left[\beta_A \frac{I_{AL}}{N_L} + \beta_C \frac{I_{CL}}{N_L}\right].
$$ 
The total per-capita transmission rate for low-risk susceptibles is given by
$$
 \phi_L = \psi_L + \tau_L.
$$

\paragraph{Enrollment} 
For the enrollment on PrEP and TaP, we assume that both is done by randomly sampling individuals at locations where high-risk individuals resp.~chronically infecteds are prevalent. In case a sampled individual turns out to be a high-risk susceptible resp.~chronically infected, he is urged to enroll in PrEP resp.~TaP.

That is, for PrEP we assume that there is a high-risk environment (HRE) where high risk individuals are overrepresented, like bars or sex clubs. The recruitment of the patience then would be carried out during the period in which the location is strongly frequented, i.e., typically in the evening on weekends. Then, denoting the probability that a random high resp.~low risk individual is in a HRE at a random moment during the recruitment period by
\begin{eqnarray*}
p_H & = & P(\mbox{HRE} | R = H)  \\
p_L & = & P(\mbox{HRE} | R = L),
\end{eqnarray*}

The probability of a random individual encountered in a HRE at a random moment to be a high-risk susceptible is
\begin{eqnarray*}
P(S_H | \mbox{HRE}) 
& = & \frac{ P(\mbox{HRE} | R = H)P(S_H) }{ P(\mbox{HRE} | R = H)P(R = H) + P(\mbox{HRE} | R = L)P(R = L) } \\
& = & \frac{ p_H\frac{S_H}{N} }{ p_H\frac{N_H}{N} + p_L\frac{N_L}{N} } = \frac{ r_bS_H }{ r_bN_H + N_L } =: \zeta_P, \\
\end{eqnarray*}
where $r_b = p_H / p_L$. That is, $r_b$ is the odds of a high-risk person to go to a HRE compared to a low-risk person, which we assume to be $r_b=0.8/0.2$ in our work. Consequently, if $u$ is the relative rate at which individuals are sampled at HREs and then put on PrEP in case they are high risk susceptibles, the absolute rate for transition from $S_H$ to $P$ is
\label{value_rb}
\[
  uN P(S_H|\mbox{HRE}) = u\zeta_PN.
\]

Using the same modeling approach for TaP, we obtain the absolute rate for transition from $I_{CH}$ to $T_H$ beyond the baseline to be $v\zeta_{T,H}N$
with
\[
  \zeta_{T,H} = \frac{ r_b I_{CH} }{ r_b N_H + N_L }.
\]
Analogously, the absolute rate for transition from $I_{CL}$ to $T_L$ beyond the baseline is $v\zeta_{T,L}N$
with
\[
  \zeta_{T,L} = \frac{ I_{CL} }{ r_b N_H + N_L }.
\]

\paragraph{Dynamical system} 
Based on the considerations presented in the previous paragraphs, we obtain the following system of ODEs describing the dynamics of our system:

\begin{equation}\label{eq:sys}\begin{aligned}
\dot{S}_H =~& \alpha_H -(\phi_H(X) + \rho_H + \mu)S_H + \rho_L S_L +xP - u_P\zeta_P(X)N\\
\dot{S}_L =~& \alpha_L - (\phi_L(X) + \rho_L + \mu)S_L + \rho_H(S_H + P)\\
\dot{I}_{CH} =~& \delta_A I_{AH}  - (\rho_H + \mu + \delta_C + v_b)I_{CH} + \rho_L I_{CL} + yT_H - u_T\zeta_{T,H}(X)N\\
\dot{I}_{CL} =~& \delta_A I_{AL}  - (\rho_L + \mu + \delta_C + v_b)I_{CL} + \rho_H I_{CH} + yT_L - u_T\zeta_{T,L}(X)N\\
\dot{I}_{AH} =~& \phi_H S_H - (\rho_H + \mu + \delta_A)I_{AH} + \rho_L I_{AL}\\
\dot{I}_{AL} =~& \phi_L S_L - (\rho_L + \mu + \delta_A) I_{AL} + \rho_H I_{AH}\\
\dot{T}_H =~& -(y + \rho_H + \mu)T_H + v_bI_{CH} + \rho_L T_L + u_T\zeta_{T,H}(X)N \\
\dot{T}_L =~& -(y + \rho_L + \mu)T_L + v_bI_{CL} + \rho_H T_H + u_T\zeta_{T,L}(X)N\\
\dot{P} =~& -(x + \rho_H + \mu)P + u_P\zeta_P(X)N,
\end{aligned}\end{equation}
where $X=\begin{bmatrix}T_H&T_L&I_{CH}&I_{CL}&I_{AH}&I_{AL}&S_H&S_L&P \end{bmatrix}'$ is the vector of state variables; $N$ is the sum of all states, $N=\langle\mathbf{1}, X\rangle$, $\mathbf{1}$ is the column of ones; $\phi_H(X)$, $\phi_L(X)$, $\zeta_{T,H}(X)$, $\zeta_{T,L}(X)$, and $\zeta_{P}(X)$ are the non-linear (rational) functions of $X$ which take on non-negative values for any $X\in \R^n_{\ge 0}$; $u_P$ and $u_T$ are the control inputs which correspond to the fraction of total population being involved either in PrEP ($u_P$) or in TaP ($u_T$), and all the remaining terms are non-negative constants.

Below, we analyze an important property of the system (\ref{eq:sys}) that will be used later on.

\subsection{Nonnegativity}

Consider the control system 
\begin{equation}\label{eq:sys-example}
\dot{X}=f(X)+g(X,u),
\end{equation}
where $f(X):\R^n\rightarrow \R^n$, and $g(X):\R^n\times\R^m\rightarrow\R^n$.

The system (\ref{eq:sys-example}) is said to be {\em nonnegative} if any solution starting at $t_0$ from $X_0\in \R_{\ge 0}^n$ belongs to $\R_{\ge 0}^n$ for all $t\in[t_0,\infty)$, i.e.,  $\forall X_0\in \R_{\ge 0}^n$, $\forall t\in[t_0,\infty)$, $X(t_0,X_0,t)\in \R_{\ge 0}^n$ (see \cite{Haddad:05,Haddad:10} for more details).

\begin{defn}Let $f = [f_1, \dots, f_n]': \R^n \rightarrow \R^n$. Then $f$ is {\em essentially nonnegative} if $f_i(X) \ge 0$, for all $i = 1, \dots, n$, and $X \in \R^n_{\ge 0}$ such that $X_i = 0$, where $X_i$ denotes the $i$-th element of $X$.\end{defn}

We have the following result:
\begin{thm}
The system (\ref{eq:sys-example}) is nonnegative for any nonnegative control $u(t):[t_0,\infty)\rightarrow \R_{\ge 0}^m$ if $f(X)$ and $g(X,u)$ are essentially nonnegative.
\end{thm}
\begin{proof}
This result can be proved geometrically by examining the direction of the vector field on the bounding hyperplanes $H_i=\{X\in\R_{\ge 0}^n|X_i=0\}$. The essential nonnegativeness condition guarantees that on each $H_i$ the system's vector field points towards the positive orthant $\R_{\ge 0}^n$.
\end{proof}

In our case, we can readily observe that $\phi_H$ and $\phi_L$ as well as $\zeta_{T,L}$, $\zeta_{T,H}$, and $\zeta_{P}$ are positive for all $X\in \R^n_{\ge 0}$. Also, we have that the latter three functions turn to zero when, respectively, $I_{CH}$, $I_{CL}$ or $S_H$ is equal to zero.  Thus the right hand side of (\ref{eq:sys}) is essentially nonnegative and the system's dynamics is non-negative either.


\section{Optimal control problem}\label{sec:OC-problem}

The optimal control problem is to minimize the total cost while respecting certain structural and budgetary restrictions. The instantaneous cost is defined as $C(t,X(t),U(t))$ and the total cost to be minimized is
\begin{equation}\label{eq:problem}
J^C(X)=\int\limits_0^{t_f} C(t,X(t),U(t)) dt,
\end{equation}
where $t_f$ is the time horizon which is chosen to correspond to the duration of the intervention. In the following we set the initial time to 0, but the problem can be easily modified to account for an arbitrary initial time. Note that a possible problem statement could include a discounting factor $e^{-\rho t}$, $\rho>0$ which is typically used to describe our priorities: one may attach greater importance to decreasing the incidence in the near future while paying less attention to what will happen in the farther future.  This factor could be used to discount some of the future benefits for shorter term benefits to account to for uncertainty in the validity of the model assumptions in the far future.

\begin{rem}
Note that the optimization problem with the cost functional (\ref{eq:problem}) is well-posed if $C(t,X(t),U(t))$ is nonnegative for all admissible values of $X(t)$ and $U(t)$. Since the system (\ref{eq:sys}) is nonnegative, an instantaneous cost function defined as a linear combination of the state variables (as it is typical in epidemiological applications) would satisfy this requirement.
\end{rem}

The restrictions imposed on the system are twofold:

\paragraph{Restriction on the admissible control policies} The first class of restrictions is due to the structural limitations of the decision unit. Since the intervention is performed by a medical organization the control profile must be sufficiently regular. We assume that the set of admissible controls consists of piecewise constant functions with a fixed interval between two consecutive switches. This restriction can be easily relaxed in different ways: we can assume that the switches occur at some non-regular times, that the control is not piece-wise constant, but rather piece-wise linear (or even piece-wise continuous) and so forth.

Let $\mathcal{T}=\{t_i\}_{i=0}^{n_{int}}$, $0=t_0<t_1<\ldots<t_{n_{int}}=t_f$ be the time instants at which control switches occur along with the initial and final time. We assume that for any $1\le i\le n_{int}$, the duration of the respective interval is constant: $t_i-t_{i-1}=\delta t$. In practice, $\delta t$ is chosen to be a multiple of 1 year. The set of admissible controls is thus:
\begin{equation}\label{eq:U-adm}\mathcal{U}:[0,t_f)\rightarrow \R_{\ge 0}^m\,\mbox{ s.t. }\, \mathcal{U}(t)=U^i\in \R^m_{\ge 0},\; t\in[t_{i-1},t_{i}),\; 1\le i\le n_{int}.\end{equation} 

The controls are assumed to be right-continuous at $t_i$. In this way, a continuous control is described by a set of its discrete values $\{U^i\}_{i=1,\ldots,n_{int}}$. The goal of the optimization is to determine these values in order to minimize the cost function (\ref{eq:problem}) while respecting the constraints. 

Note that the controls are bounded by zero from below, but there are no upper bounds. The upper bounds are imposed implicitly as will be described below.

\paragraph{Dynamic budget allocation} The second class of constraints is due to the budgetary limitations. In practice, an intervention incurs large expenses which are compensated by the government only to some extent. We assume that at the beginning of each control interval $[t_{i-1},t_{i})$ the government sets a baseline budget by estimating the expected expenses and allocates the money to be spent for the intervention starting from this baseline. This works as follows:

The total expenditures related to treating people with TaP or PrEP and the enrollment costs for TaP or PrEP are captured by the following cost function:  
$$J^B_i(X,U^i)=\int\limits_{t_{i-1}}^{t_{i}} B^i(X(s),U(s)) ds,$$
where $B^i(X(s),U(s))$ are certain positive defined functions. In our case, we formulate the optimal control problem to minimize the incidence of HIV infection. The instantaneous cost is thus defined as the incidence rate of HIV, i.e., 
$$C(X,U)=S_H\phi_H(X)+S_L\phi_L(X).$$ 
Furthermore, when computing the budgetary restrictions we assumed   
$$B^i(X(t),U^i)=K^{(t)}_T [ T_H(t) +T_L(t)] + K^{(t)}_P P(t) + K^{(e)}_T N(t)u^i_T(t) + K^{(e)}_P N(t) u_P^i(t)$$
for $i=1,\ldots,n_{int}$ and $t\in [t_{i-1},t_i)$. Here, $K^{(t)}_T$ and $K^{(t)}_P$ are the monthly cost for treatment with TaP and PrEP, respectively, per patient. The coefficients $K^{(e)}_T$ and $K^{(e)}_P$ represent the costs for approaching and, if necessary, enrolling one patient into TaP and PrEP, respectively.

We compute this cost for the uncontrolled case to determine the baseline expenses, i.e., the expenses that the government would defray if there is no intervention. The assigned budget is allocated atop the baseline budget. Let $\tilde X_i(t)$, $t\in[t_{i-1},t_{i})$ be the uncontrolled ($U^i=\mathbf{0}$) solution of (\ref{eq:sys}) with initial condition $\tilde X_i(t_{i})=X(t_{i})$. The dynamic budget constraint (DBC) is thus formulated as follows:
\begin{equation}\label{eq:B}J^B_i(X,U^i)-J^B_i(\tilde X_i,\mathbf{0})\le B,\quad i=1,\dots,n_{int}, \end{equation}
where $U^i \in \R^m_{\ge 0}$ and 
$$J^B_i(\tilde X_i,\mathbf{0})=\int\limits_{t_{i-1}}^{t_{i}} B^i(X(s),\mathbf{0})ds.$$
The constraints (\ref{eq:B}) set an implicit limit to the set of admissible controls $\mathcal{U}$. In terms of optimal control theory such constraints can be classified as {\em mixed integral inequality path constraints}. There is in general no way to handle such constraints analytically, but they can be treated numerically as will be shown below.

Finally we formulate the resulting optimization problem as follows. Determine the discrete values of control $U^i\in \R^2_{\ge 0}$, $i=1,\ldots,n_{int}$ s.t.
\begin{equation}\label{eq:OC-problem}\begin{cases}J^C(X)=\int\limits_0^{t_f} C(t,X(t),U(t)) dt \rightarrow \min\\
J^B_i(X,U^i)-J^B_i(\tilde X_i,\mathbf{0})\le B, i=1,\ldots,n_{int}\\[2pt]
X(t), \,t\in[0,t_f],\quad \mbox{ satisfies } (\ref{eq:sys}) \mbox{ with } X(0)=X_0 \mbox{ and } u(t)=U^i, t\in[t_{i-1},t_i),\\[2pt]
\tilde X_i(t), \,t\in[t_{i-1},t_i]\, \mbox{ satisfies } (\ref{eq:sys}) \mbox{ with } X_i(0)=X(t_i) \mbox{ and } u(t)=0, t\in[t_{i-1},t_i).
\end{cases}\end{equation}

\section{Numerical solution of the optimal control problem}\label{sec:OC-num}

Consider a dynamic system (\ref{eq:sys-example}). From now on we will follow the established convention and will assume that the state $X(t)$ is a row vector. Both $f(X)$ and $g(X,u)$ are thus row-valued vector functions.

A collocation method interpolates the state and the control functions at a number of time points (called grid points), and requires the solution to satisfy the respective differential equation at the collocation points which may not necessarily coincide with the grid points. Some grid points can be used to ensure additional conditions on the solution, e.g., continuity.


\subsection{Lagrange interpolation}
The optimal state trajectory is a piece-wise smooth function whose first derivative is discontinuous at the points $t_i$. Therefore it is natural to break it into $n_{int}$ intervals coinciding with $[t_{i-1},t_i)$ and interpolate on each interval separately using the basis of Lagrange polynomials, \cite{BerTre:04}
$$L^i_k(t)=\prod\limits_{l=0,\,l\neq k}^{n_{cp}} \frac{t-\tau^i_l}{\tau^i_k-\tau^i_l},$$
where $n_{cp}$ is the number of collocation points $\tau^i=\{\tau^i_k\}_{k=0,\dots,n_{cp}}$ within the $i$th interval\footnote{In the following, the upper index ${}^i$ will refer to the number of the respective interval $[t_{i-1},t_i)$, $i=1,\ldots,n_{int}$. We will drop this superscript when the reference to a specific interval is not relevant. Furthermore, the lower index ${}_j$ will refer to the element of the state vector, i.e., $j=1,\ldots,n$.}, $i=1,\dots,n_{int}$. Since the intervals are of equal length we assume that the number of grid points is the same for each interval. 

Figure \ref{fig:Lag_poly} shows a family of Lagrangian polynomials defined on a non-uniform grid. Notice that for any grid point there is only one polynomial that takes on a non-zero value (which is equal to 1) at this point.

\begin{figure}[ht]
  \centering
    \includegraphics[width=0.8\textwidth]{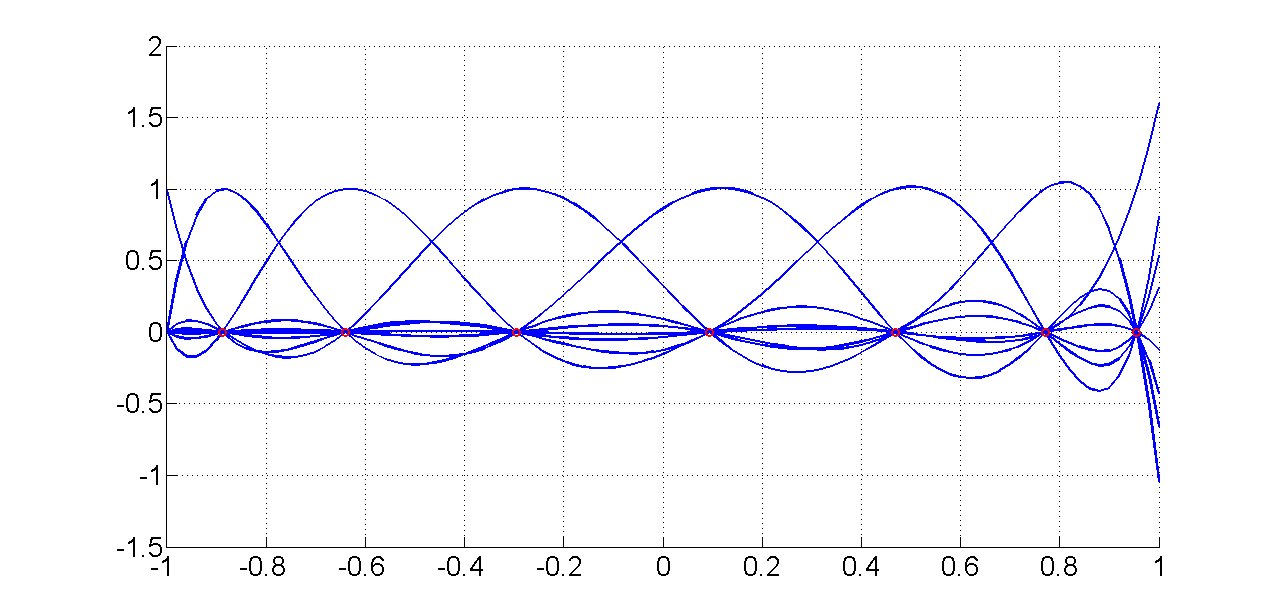}
    \caption{A family of Lagrange polynomials on the interval $[-1,1]$. The grid points are indicated by red circles on the $x$ axis.}
    \label{fig:Lag_poly}
\end{figure}

Let $X(t)$ be the state trajectory for $t\in[t_{i-1},t_i)$. We define $\mathbf{X}^i$ to be the $[(n_{cp}+1)\times n]$ matrix of the values of the state at times $\tau_k^i$, i.e.,  $\mathbf{X}^i_{kj}=[X_j(\tau^i_k)]$, $j=1,\dots,n, k=0,\dots,n_{cp}$. The interpolating polynomial for the $j$th component of the state over the $i$th interval is thus 
$$\hat{X}^i_j(t)=\sum\limits_{k=0}^{n_{cp}} L^i_k(t) X_j(\tau^i_k).$$

It is well known that for a regular (i.e., equispaced) grid the interpolating polynomial may fluctuate heavily between the interpolation points, especially close to the endpoints of the interval (this is referred to as the {\em Runge phenomenon}, see, e.g., \cite{Gau:11}). To overcome this drawback one uses unevenly spaced grid points whose distribution density increases as we approach the endpoints of the interval. There are two standard choices for the grid points: Legendre and Chebyshev points which are zeros of Legendre or Chebyshev polynomials. These polynomials belong to the class of orthogonal polynomials thus giving the name to the method ({\em orthogonal collocations}), \cite{Gau:11,Can:06}. 

\begin{figure}[ht]
  \centering
\hspace*{-2.5em}\includegraphics[width=0.58\textwidth]{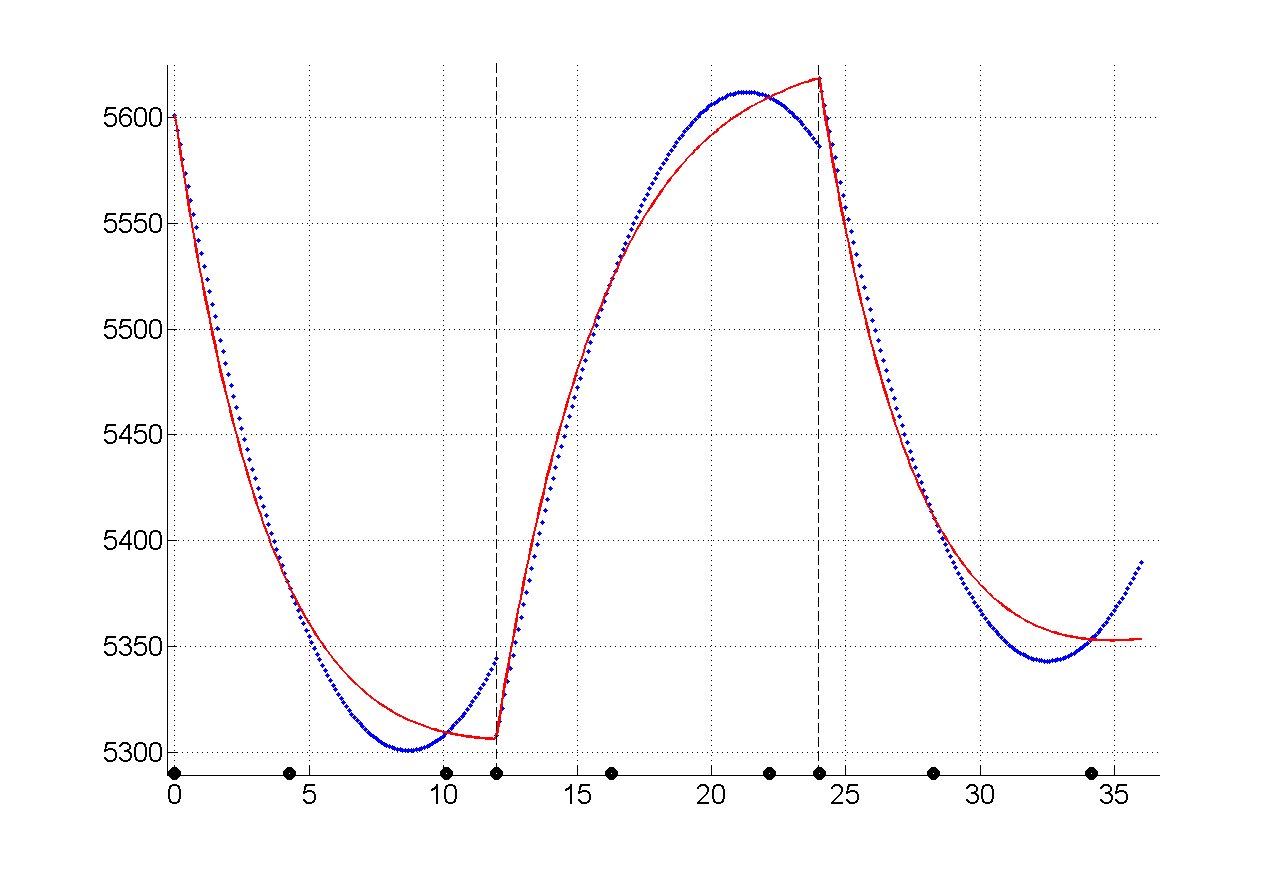}\hspace*{-3em}\includegraphics[width=0.58\textwidth]{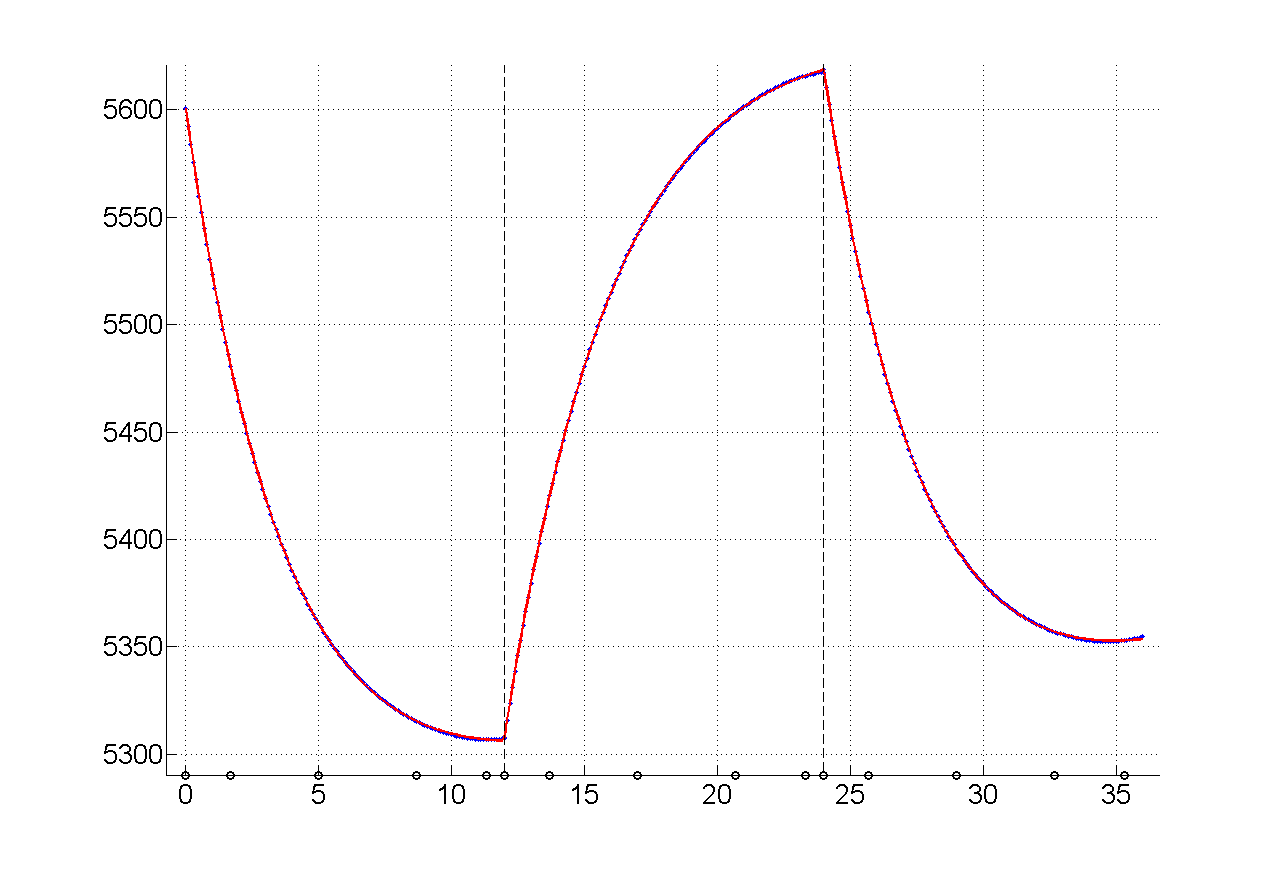}

\vspace*{-1em}a)\hspace*{25em}b)
%
    \caption{Interpolation of a sample trajectory using different number of non-uniformly spaced grid points: a) 3 points, b) 5 points. The sample trajectory is shown in blue, its interpolation is in red. Grid points are indicated by circles on the $x$-axis. Dashed lines separate different intervals.} 
    \label{fig:approx1}
\end{figure}

Figure \ref{fig:approx1} illustrates this thesis. It shows a sample trajectory along with its Lagrange interpolation obtained using 3 and 5 non-uniformly distributed points per interval. It is seen that the interpolation becomes nearly exact already with 5 grid points. In contrast to this, when using regular grid points the interpolating function may deviate from the approximated trajectory by several orders of magnitude.

There has been an extensive discussion regarding the merits and drawbacks of either of two choices (see, e.g., \cite{Tre:08} and references therein). However, it seems that neither of the two is clearly superior to another one. In this study, we have chosen to use the Legendre points.

\subsection{Collocation at Legendre-Gauss-Radau points}

In the previous step the state was described using Lagrange interpolation polynomials which go through $(n_{cp}+1)$ values of the state $X_j(t)$ at $t=\tau^i_k$. Determining the values $X_j(t)$ corresponds to determining an approximation of the solution of (\ref{eq:sys-example}). 
 To do so the derivatives of the interpolating polynomials are computed at points $\tau^i_k$, $k=1,\ldots,n_{cp}$, i.e., $\dot{\hat{X}}_j^i(\tau^i_k)$ and the computed derivatives are required to {\em collocate} with the right-hand sides of (\ref{eq:sys-example}) computed for $X_j(\tau^i_k)$. Additionally, the continuity of the trajectory can be ensured by requiring the interpolating polynomials to be attached to each other at the knot points $\tau^i_0=t_i$, $i=1,\ldots,(n_{int}-1)$: $\hat{X}^{i-1}(t_i)=\hat{X}^i(t_i)$.

To get the required distribution of points we use a particular class of Legendre points, referred to as the Legendre-Gauss-Radau (LGR) points, which are the roots of $P_{n_{cp}+1}(\theta) + P_{n_{cp}}(\theta)$, where $P_{k}(\theta)$ is the $k$-th degree normalized Legendre polynomial, $P_k(1)=1$. The LGR points have the property that one of these points coincides with the left endpoint of the interval, i.e., the respective polynomial has a root at $\theta=-1$\footnote{One can also define LGR points which include the right endpoint instead the left one.}. That is, the LGR points are defined on the interval $[-1,1)$.

\begin{rem}Note that alternatively one can use the Legendre-Gauss (LG) or the Legendre-Gauss-Lobato (LGL) points which are determined as the roots of $P_{n_{cp}}(\theta)$ or as the roots of $\dot{P}_{n_{cp}-1}(\theta)$ together with $\{-1\}\cup\{1\}$. The LG points lie completely within the interval while the LGL points include both endpoints. That is, the LG are defined on $(-1,1)$ and the LGL points on $[-1,1]$.
\end{rem}

Any $t\in [t_{i-1}, t_{i})$ can be associated with $\theta\in [-1,1)$ by the affine transformation 
\begin{equation}\label{eq:t-transform}t=t_{i-1}+\frac{\delta t}{2}(\theta+1).\end{equation}

Consider the $(n_{cp}+1)$ LGR points $\theta_k$, $k=0,\dots,n_{cp}$ with $\theta_0=-1$. One can map the LGR points $\theta_k$ to the respective grid points $\tau^i_k$ using (\ref{eq:t-transform}), i.e., $\tau_k^i=t_{i-1}+\frac{\delta t}{2}(\theta_k+1)$. Following the procedure described above we use Lagrange polynomials to approximate the state trajectory at points $\tau^i_k$. For the $j$th component of the state this results in a polynomial of degree at most equal to $n_{cp}$:
\begin{equation}\label{eq:XN}X_j(t)\approx\hat{X}^i_j(t)=\sum\limits_{k=0}^{n_{cp}} X_j(\tau^i_k) L^i_k(t),\quad t\in[t_{i-1},t_i).\end{equation}
Differentiating (\ref{eq:XN}) and evaluating at the collocation point $\tau^i_k$ we get
\begin{equation}\label{eq:dXN}
\dot{\hat{X}}^i_j(\tau^i_k)=\sum\limits_{l=0}^{n_{cp}} X_j(\tau^i_l) \dot{L}^i_l(\tau^i_k)=\sum\limits_{l=0}^{n_{cp}} X_j(\tau^i_l) D^i_{kl}  ,\end{equation}
where $D^i$ is an $[n_{cp}\times (n_{cp}+1)]$ {\em differentiation matrix} whose $(k,l)$-th element is the derivative of the Lagrange polynomial $L^i_l$ at the collocation point $\tau^i_k$, $k=1,\dots,n_{cp}$. Note that we do not collocate at the knot points $\tau_0^i=t^i$. The differentiation matrix $D$ extends the idea of finite difference approximations of derivatives which are computed based on the trajectory evaluation at a finite number of points. In our case, the approximation of the derivative is a function of the state at all the $n_{cp}+1$ grid points. 

\begin{rem}
In \cite{Garg:10}, it was shown that the differential matrix $D$ turns out to be singular when the number of collocation points is equal to the number of grid points (here denoted by $n_{gp}$). The reason for this is that while the interpolating polynomial is degree $n_{gp}-1$ its derivative is degree $n_{cp}-2$ and so, requires only $n_{gp}-1$ conditions to be uniquely determined. Thus the number of collocation points has to be one less than the number of grid points.

To overcome this difficulty, it is proposed in \cite{Garg:10} to use the set of LGR or LG points along with a boundary point $\{-1\}$ or $\{1\}$. The interpolating polynomial is computed for the extended set of points while the collocation is carried out only at the LGR, resp. LG points. 

The flaw of this approach is that it results in a set of grid points which is not longer produced by an orthogonal polynomial. This leads to a decrease in the accuracy of the resulting polynomial interpolation. The obvious remedy is to compute the interpolating polynomial for the whole set of orthogonal points while collocating at all but one point. The remaining point can be used to enforce the continuity condition as it is done in this paper.
\end{rem}

The Lagrange basis polynomials defined for the LGR points can be written in barycentric form \cite{BerTre:04} as
$$L_j(t)=\frac{P_{n_{cp}+1}(t)+ P_{n_{cp}}(t)}{(t-\tau_j)\big(\dot{P}_{n_{cp}+1}(\tau_j) + \dot{P}_{n_{cp}}(\tau_j)\big)},$$
whence (see \cite{Shen:11} for the derivation)
$$\dot{L}_j(\tau_k)=\begin{cases}
\frac{\displaystyle\dot{P}_{n_{cp}+1}(\tau_k)+ \dot{P}_{n_{cp}}(\tau_k)}{\displaystyle(\tau_k-\tau_j)\big(\dot{P}_{n_{cp}+1}(\tau_j) + \dot{P}_{n_{cp}}(\tau_j)\big)},& j\neq k,\\[7pt]
\frac{\displaystyle\tau_k}{\displaystyle 1-\tau_k^2}\frac{\displaystyle(n_{cp}+1)P_{n_{cp}(\tau_k)}}{\displaystyle(1-\tau_k^2)\big(P_{n_{cp}+1}(\tau_k)+ P_{n_{cp}}(\tau_k)\big)},&j=k\neq 0,\\[7pt]
-\frac{\displaystyle n_{cp}(n_{cp}+2)}{\displaystyle 4},&j=k=0.\end{cases}$$ 

\begin{rem}
Note that the Lagrange basis polynomials are invariant with respect to the shift or the dilatation of the abscissa axis. Thus the approximate differentiation matrix does not change for different intervals  provided the number and the type of grid points do not change. However, the use of the affine transformation (\ref{eq:t-transform}) implies that the system's differential equations should be modified accordingly to take into account the transformed time variable. In practice, this means that the right-hand sides of (\ref{eq:sys-example}) are to be multiplied with the correction factor $\delta t/2$.
\end{rem}

Denoting by $\dot{\hat{\mathbf{X}}}^i_{kj}$ the approximate values of the derivatives of $X_j(t)$ at collocation points $\tau_k^i$ we write compactly $\dot{\hat{\mathbf{X}}}^i= D \mathbf{X}^i$. We also compute the derivatives of the state trajectory by evaluating the right-hand side of the ODE (\ref{eq:sys-example}). We write $F(\hat{\mathbf{X}}^i,U^i)_{kj}=f_{j}(X(\tau^i_k))+g_{j}(X(\tau^i_k),U^i)$ and the resulting set of $n_{cp}\cdot n$ constraints is hence 
\begin{equation}\label{eq:Constr-F}
D\mathbf{X}^i-\frac{\delta t}{2}F(\mathbf{X}^i,U^i)=0.
\end{equation}

\subsection{Defect constraints}

While (\ref{eq:Constr-F}) determine values of $n\cdot n_{cp}\cdot n_{int}$ samples of the state trajectory, there are still $n\cdot n_{int}$ free $X$'s which can be used to ensure continuity of the trajectory. These are the values of $X$ at the knot points $t=t_i$, $i=0,\dots,n_{int}-1$. The initial value $X_0$ determines the first $n$ points: $X(t_0)=X_0$. The remaining values should be determined from the defect constraints which can be written as follows:
$$X^{i}(t_i)-X^{i+1}(t_i)=0\quad \forall i=1,\dots,n_{int}-1.$$
The second term in the left-hand side is $X^{i+1}(t_i)=X(\tau^i_0)$ and the first term is 
\begin{equation}\label{eq:X-int}X^{i}(t_i)=X^{i}(t_{i-1})+\int_{t_{i-1}}^{t_i} f(X(s))+g(X(s),U^i)ds.\end{equation}
The integral in (\ref{eq:X-int}) can be computed numerically using {\em Gaussian quadrature}:
\begin{equation}\label{eq:G-quad}\int_{t_{i-1}}^{t_i} f_{j}(X(s))+g_{j}(X(s),U^i)ds=\frac{\delta t}{2}\sum\limits_{k=0}^{n_{cp}}w_k F_j(\mathbf{X}(\tau^i_k),U^i),\end{equation}
where $w_j$ are the weights computed for the given distribution of grid points, \cite{Can:06,Gau:11}. 

Now we have $n\cdot (n_{cp}+1)\cdot n_{int}$ constraints to determine the same number of discrete values of the state function. Solving these equations is equivalent to getting a numerical solution of the respective differential equations. In other words, we reduced the procedure of integrating the system of ODEs to solving a system of nonlinear algebraic equations.

Now we have the solution for given values of the control $\{U^i\}$, $i=0,\dots,n_{int}-1$. Before proceeding to the optimization we have to formulate the constraints imposed on the controls.

\subsection{Numerical approximation of dynamic budget constraints}

To compute the dynamic budget constraints one has to solve the system's ODEs with zero controls for $n_{int}$ intervals with initial conditions determined by the solution of the controlled system. This problem can be formulated within the considered framework in the following way. Let $X_0^i$ be state values at the grid points within the $i$th interval. We set $\mathbf{X}^i_0(t_{i-1})=\mathbf{X}(t^i_{i-1})$ and the remaining discrete values $\mathbf{X}^i_0(\tau_k^i)$, $k=1,\ldots,n_{cp}$ are determined from 
\begin{equation}\label{eq:Constr-F0}
D\mathbf{X}_0^i-\frac{\delta t}{2}F(\mathbf{X}_0^i,0)=0.
\end{equation}
The resulting set of inequality constraints is thus
$$\frac{\delta t}{2}\sum\limits_{k=0}^{n_{cp}}w_k \left[B^i(\mathbf{X}^i(\tau^i_k),U^i)-B^i(\mathbf{X}^i_0(\tau^i_k),0)\right]-B_{lim}\le 0,\quad i=1,\ldots,n_{int}.$$

\subsection{Constrained nonlinear programming problem}

The cost function (\ref{eq:problem}) is computed numerically using Gaussian quadrature and hence the optimal control problem (\ref{eq:OC-problem}) turns into the following constrained optimization problem:
\begin{equation*}\begin{cases}
\dfrac{\delta t}{2}\sum\limits_{i=1}^{n_{int}}\sum\limits_{k=0}^{n_{cp}}w_k C(\mathbf{X}(\tau^i_k),U^i)\rightarrow \min\\[7pt]
\mbox{s.t. } D\mathbf{X}^i-\dfrac{\delta t}{2}F(\mathbf{X}^i,U^i)=0\\[3pt]
\qquad X(\tau^i_0) - \dfrac{\delta t}{2}\sum\limits_{k=0}^{n_{cp}}w_k F_j(\mathbf{X}(\tau^i_k),U^i)=0,\; i=1,\ldots,n_{int},\\[3pt]
\qquad D\mathbf{X}_0^i-\dfrac{\delta t}{2}F(\mathbf{X}_0^i,0)=0,\hspace*{6em} i=1,\ldots,n_{int},\\[3pt]
\qquad \mathbf{X}^i_0(t_{i-1})=\mathbf{X}(t^i_{i-1}),\hspace*{7.5em} i=1,\ldots,n_{int},\\[3pt]
\qquad\dfrac{\delta t}{2}\sum\limits_{k=0}^{n_{cp}}w_k \left[B^i(\mathbf{X}^i(\tau^i_k),U^i)-B^i(\mathbf{X}^i_0(\tau^i_k),0)\right]-B_{lim}\le 0,\\
\hspace*{18em}\quad i=1,\ldots,n_{int}.
\end{cases}\end{equation*}

The above nonlinear optimization problem was implemented in Matlab with the use of {\tt fmincon} function. The optimization was performed using Sequential quadratic programming (SQP) algorithm, see Sec.\ \ref{sec:res_num} for details.

\section{Results}\label{sec:results}

We calculated the optimal allocation strategy for an outbreak scenario in a large US city with 100,000 at-risk MSM individuals. That is, the introduction of HIV into a subpopulation with a low prevalence (e.g. people aged 15-25) is considered. We assume the infection probability to be $\beta_A$ = 0.015 and $\beta_C$ = 0.001, i.e., acutely infecteds are 15 times more contagious than chronically infecteds \cite{bellan_reassessment_2015}. Furthermore, we consider a situation in which risk is static, that is, individuals do not change their risk behavior over time, $\rho_H = \rho_L = 0$. Hereby, 90 \% of the MSM population exhibits a low risk behavior, the remaining 10 \% display high-risk behavior. Hence, we set $\alpha_L=250$ to be nine times $\alpha_H=28$ so that the probability that a newly entering individual is high-risk is 10\%. 

We set $\mu = \frac{1}{360}$, leading to the individuals staying 30 years in the system on average if there was no HIV-related removal. Here removal represents either natural death, becoming sexually inactive for medical or social reasons, or settling in a monogamous relationship of two uninfected individuals. There is little data available on the mixing patterns among high and low-risk individuals that we could use to fix $\pi$. However, we assumed that $\pi=0$ which is a common assumption when mixing dynamics are unknown. Furthermore, we assume that individuals from the high risk group have ten times more sexual contacts than the ones from the low risk group based on analysis of longitudinal sexual behavioral data \cite{romero-severson_dynamic_2015}. Finally, we optimized over values of $\lambda_L$ and $\bar{u}_T$ such that at equilibrium the prevalence was 20\% and the proportion of infected individuals on treatment was 25\%, which is consistent with measured values \cite{_prevalence_2010, rosenberg_modeling_2014}, leading to $\lambda_H$ = 40.9, $\lambda_L$ = 4.09, and $\bar{u}_T$ = 0.00148. We also assume that treatment never fails and is never stopped, i.e., $y = 0$.

We consider four scenarios varying in the value of $x$: 0, $\frac{1}{60}$, $\frac{1}{24}$, $\frac{1}{12}$. That is, in one scenario PrEP never fails and is never canceled and in three scenarios this happens on average after 5, 2, and 1 year, respectively. These scenarios are meant to correspond to different distribution policies. More precisely, we assume that PrEP never fails and that a high risk individual, once identified, is prescribed PrEP either indefinitely or is only provided with it for 5, 2, and 1 year, respectively. The parameters of the cost function are given in table \ref{tab:cost}. The cost of enrollment was estimated by considering the total cost per enrollee of a similar intervention program implemented by the New York City Department of Health \cite{Blank:05} plus the cost of the labs involved in determining infectious status. The cost of TaP and PrEP is taken from \cite{alistar_effectiveness_2014}. 

\begin{table}[htb]
\centering\begin{tabular}{c | c | c | c}
  \hline			
  $K^{(t)}_T$ & $K^{(t)}_P$ & $K^{(e)}_T$ & $K^{(e)}_P$\\
  \hline
  1299 & 776 & 266 & 213\\
  \hline  
\end{tabular}
\caption{Numerical values of the coefficients in the budget functional.\label{tab:cost}}
\end{table}

The trajectories of the controls are shown in Fig.~\ref{fig:cntrl}, the ones of the number of individuals in the various states in Fig.~\ref{fig:res1}-\ref{fig:res3}. We can observe that the lower the value of $x$ is, the lower is the total number of newly infecteds and the higher is both the overall values of $u_P$ relative to $u_T$ and the number of individuals on PrEP. These two dependencies are to be expected as we try to optimize the allocation strategy of the agency in charge of fighting HIV in the considered city and this agency does not cover the long-term cost of prescribing PrEP, but only the enrollment costs. Therefore, a lower $x$ should render the overall intervention more effective and the agency should become more prone to employ PrEP more intensively. 

\begin{figure}[th]
  \centering
    \includegraphics[width=0.8\textwidth]{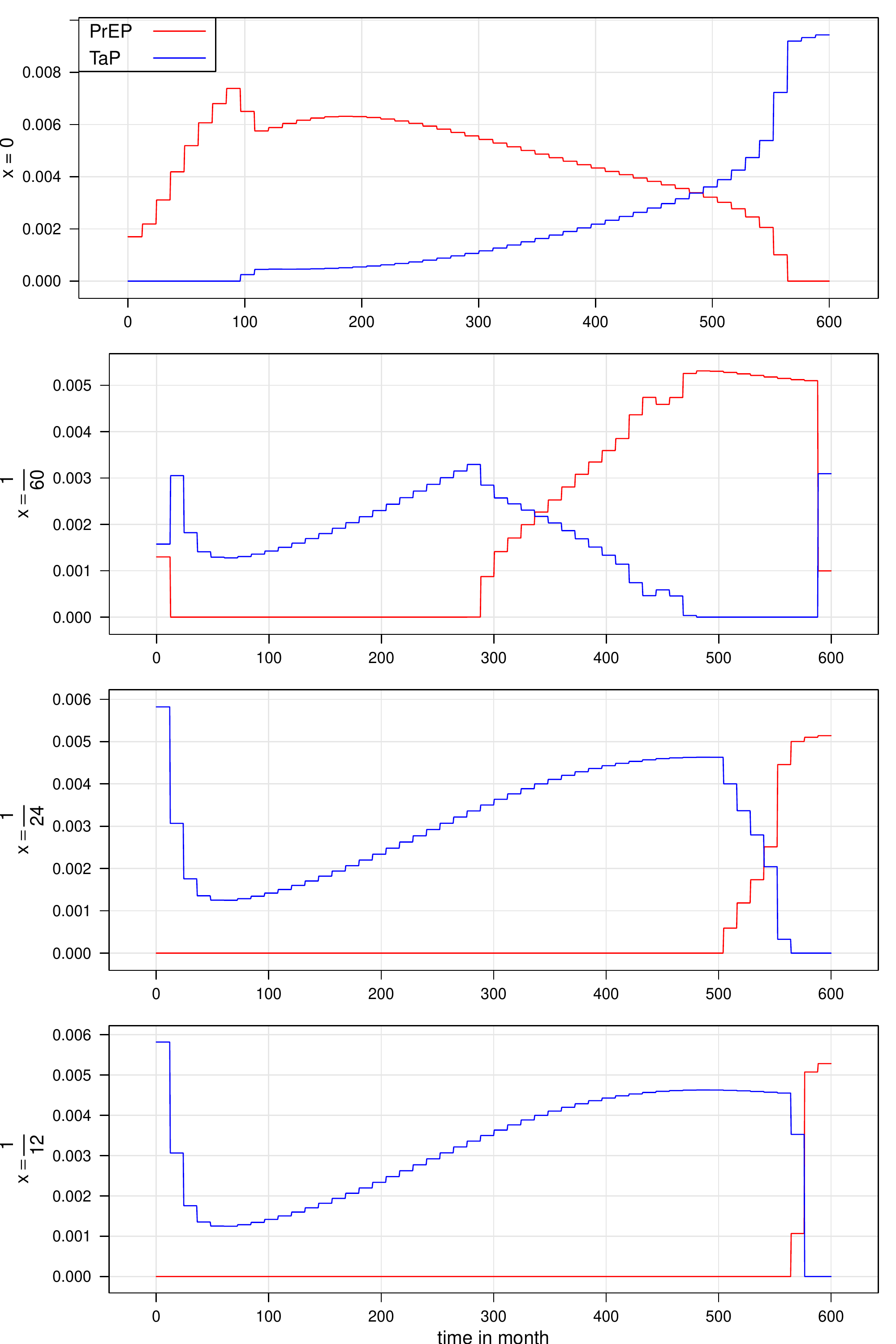}
    \caption{The trajectories of the controls for $x = 0$, $x = \frac{1}{60}$, $x = \frac{1}{24}$, $x = \frac{1}{12}$ over 50 years.}
    \label{fig:cntrl}
\end{figure}

For $x=0$, elimination is nearly achieved, whereas for the other values of $x$ this is not the case. That is, only the value $x=0$ allows the agency to push the epidemic over the tipping point where the intervention leads to the infection cycle to break down. Moreover, the difference in the number of infecteds over time is much more pronounced between the scenario with $x=0$ and the three scenarios with $x > 0$ than between the three scenarios with $x > 0$. Most extreme, the difference in the total number of infecteds is only 0.2 \% between $x=\frac{1}{24}$ and $x=\frac{1}{12}$. This is probably due to buffering effects also observed for other exogenous variables \cite{Powers:14} leading to a low influenceability of the system for these values of $x$.

Since only the consequences unfolding over the 50 year period considered are taken into account by the cost function, the choice of how to allocate the resources to enrollment into TaP and PrEP, respectively, becomes myopic to the end of the considered period. That is, for $x=\frac{1}{24}$ and $x=\frac{1}{12}$ resources are jammed into PrEP towards the end although letting the number of infecteds on treatment drop would ultimately backfire, i.e., leading to number of infecteds higher than necessary after a while. The same effect with the roles of TaP and PrEP switched can be observed for $x = \frac{1}{60}$. Moreover, since the number of infecteds or high-risk susceptibles declines over time when TaP or PrEP is favored, successful enrollment into TaP or PrEP, respectively, becomes more and more expensive, eventually favoring the other treatment. This can be observed for $x = 0$, $x = \frac{1}{60}$, and $x = \frac{1}{24}$. For $x = \frac{1}{12}$, TaP is favored over PrEP the whole period except for the switch at the very end.

It is therefore advantageous to modify the optimization problem to avoid such unrealistic results. However, as this reformulation is not always feasible, a possible remedy would be to to discard a period long enough at the end of the time interval for which the controls were optimized. 

There are two heuristic approaches to accomplish this: (a) Calculating the costs for an only-PrEP and an only-TaP strategy and determining when the last switch between these strategies in terms of which strategy incurs less costs (normally there is only one) takes place. The length of the period discarded is then chosen to be equal to the duration until this switch. (b) Running the optimization procedure on intervals of different length and checking whether the same type of strategy switch taking place at the same distance to the end of the respective calculation period occurs. If yes, this switch is an numerical artifact and needs to be discarded.

\begin{figure}[th]
  \centering
    \includegraphics[width=0.8\textwidth]{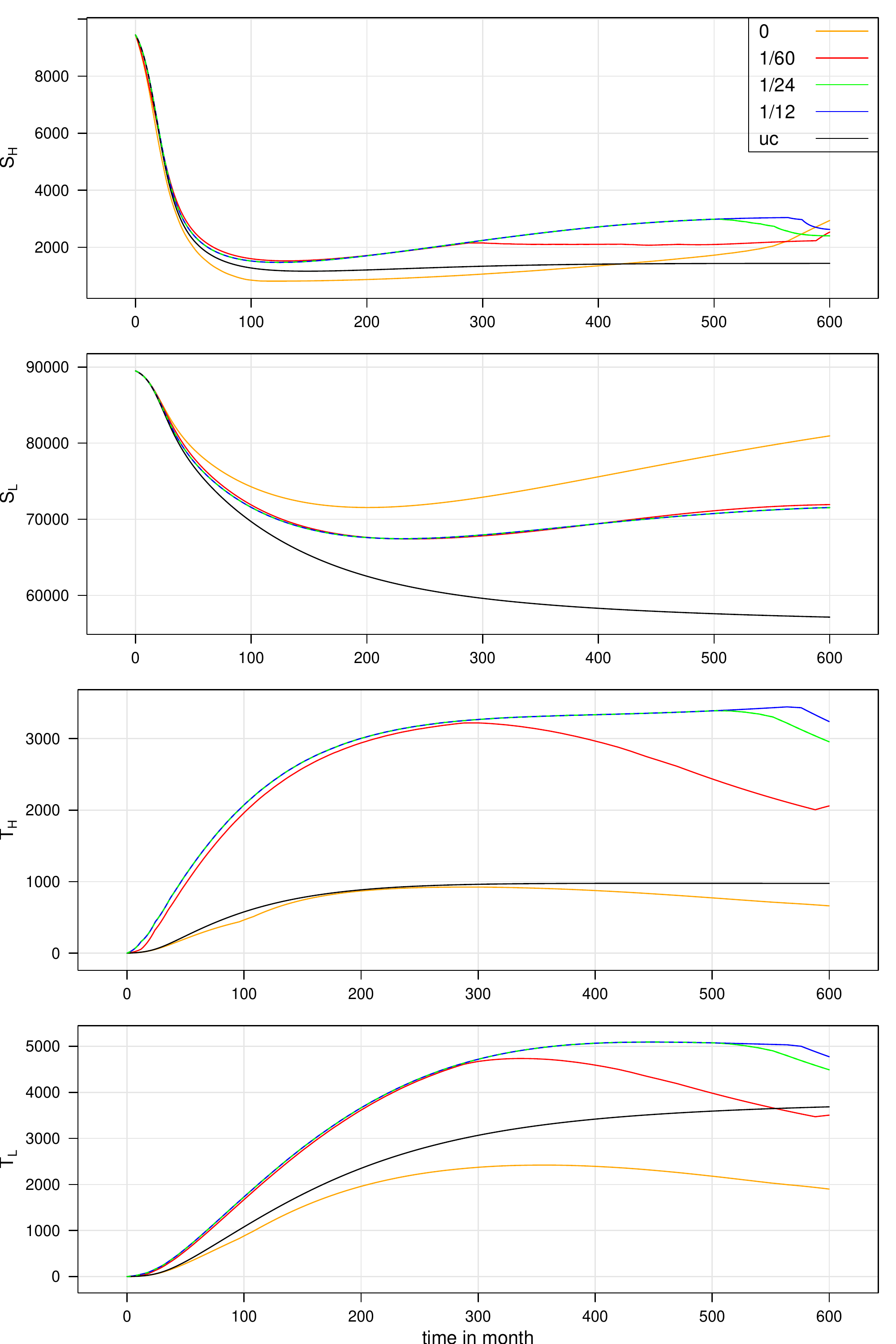}
    \caption{The trajectories of $S_H$, $S_L$, $T_H$, and $T_L$ over 50 years.}
    \label{fig:res1}
\end{figure}

\begin{figure}[th]
  \centering
    \includegraphics[width=0.8\textwidth]{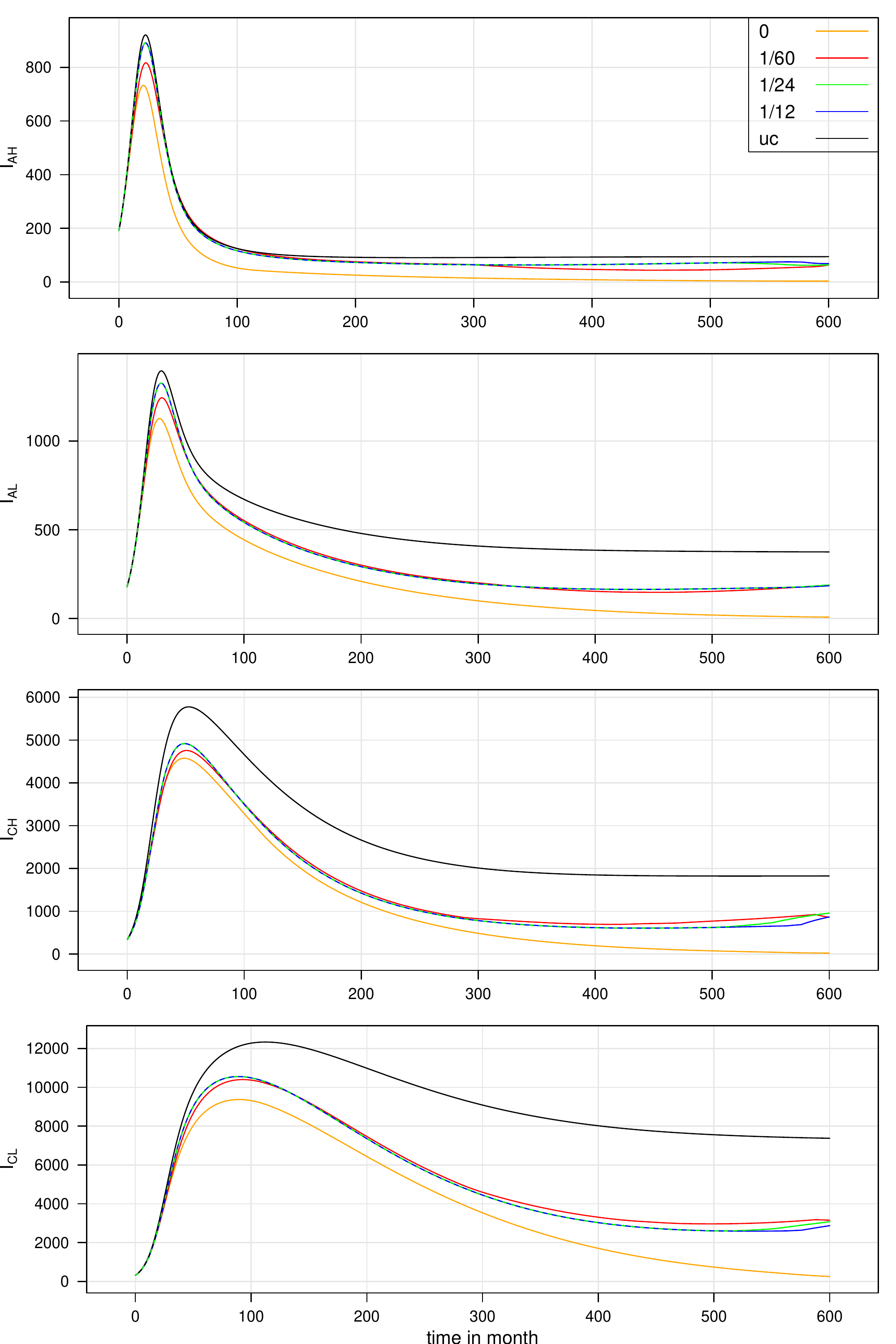}
    \caption{The trajectories of $I_{AH}$, $I_{AL}$, $I_{CH}$, and $I_{CL}$ over 50 years.}
    \label{fig:res2}
\end{figure}

\begin{figure}[th]
  \centering
    \includegraphics[width=0.8\textwidth]{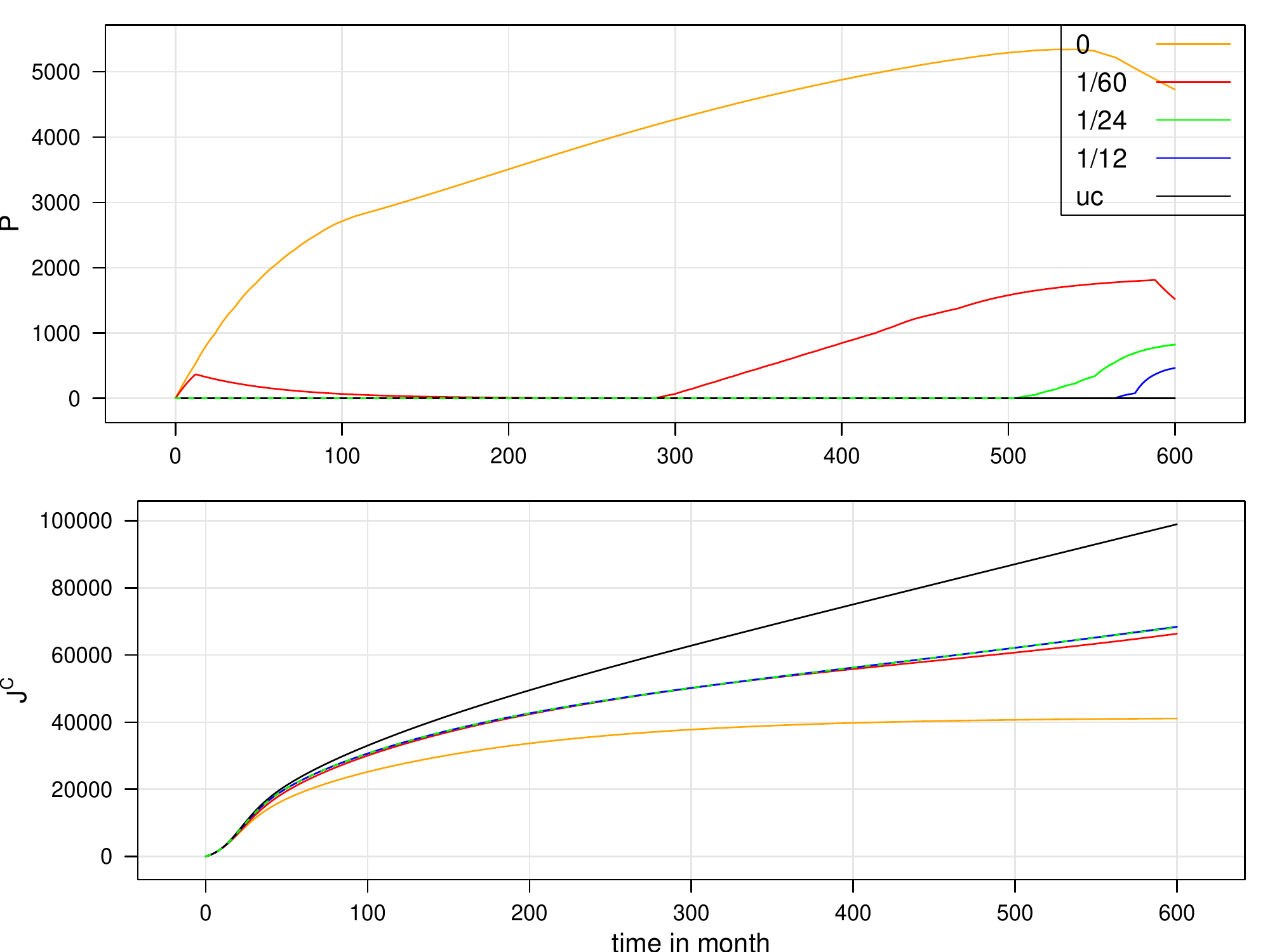}
    \caption{The trajectories of $P$ and $J^C$ over 50 years.}
    \label{fig:res3}
\end{figure}

\subsection{Numerical implementation}\label{sec:res_num}

The described optimization problem was solved both using the multiple shooting (not described in the paper) and the orthogonal collocation methods. In general, the orthogonal collocations approach is about 3-5 times faster than the multiple shooting one\footnote{For the optimization time interval equal to 50 years the computational time for orthogonal collocations was  typically within the range of $5e{+}3 - 10e{+}3$ seconds while for multiple shooting this time was about $2e{+}4 - 5e{+}4$ seconds. Note that the computational time is strongly influenced by the length of the optimization interval and by the choice of the initial guess.}. However, the latter typically provides a better result albeit the improvement never exceeds a fraction of percent. Furthermore, with orthogonal collocations one may sometimes experience the situation when the program runs out of memory. This is overcome by a slight change of the initial guess.

Practice shows that it is in general advantageous to choose an initial guess that provides a low value of the cost function while violating constraints. SQP algorithm recovers from the constraints violation in a couple of steps while keeping the cost function relatively small. If on the contrary, one chooses the initial guess in order to satisfy the constraints, the deviation from the optimal value of the cost function may turn out to be rather large. As the convergence of the algorithm is rather slow, it may take unnecessary many steps to achieve the optimal value.

The slow convergence of the algorithm is explained by the particular non-local structure of the budget constraints. Since these constraints are computed along the to-be-optimized trajectory, any change of the control variables leads to the variation of the trajectory and thus to the re-computation of the constraints. As an example, a little change of the control during the first interval influences the budget constraints all over the whole optimization interval.


To speed up the computation the gradient of the cost function was supplied to the optimization algorithm. The Hessians were computed numerically using centered finite differences due to the complexity of the respective analytical expressions.

\section{Conclusions}
This paper presents a novel model describing an HIV propagation dynamics for a geographically concentrated MSM population, along with two control actions. The two controls represent allocation of funding to the two major drug-based interventions available for HIV, HAART / TaP and PrEP, by an health agency. Addressing this setting is of particular interest as it is heavily debated to which extent PrEP should enter the intervention portfolio of health care systems \cite{Punyacharoensin:16, Baeten:12, Mugwanya:13, van_de_Vijver:13}. A suitably modified orthogonal collocations method is applied to compute optimal control profiles for a realistic outbreak scenario in a large US city. Hereby, the influence of the effective cost of PrEP on the optimal allocation strategy and the dynamics of the epidemic is studied, by varying a parameter governing said cost. The obtained results show that the allocation pattern heavily depends on the cost of PrEP, rendering PrEP the dominant intervention or not employed at all depending on said cost. Moreover, whether elimination of HIV in the considered population is achievable also is dependent on the effective cost of PrEP. 

Currently, we work on improving the developed algorithm in order to introduce it into the epidemiological community. A manuscript aimed at a public health health audience illustrating use of these methods is currently prepared.

\section*{References}


\end{document}